\documentclass[8pt]{article}
\usepackage{graphicx}
\usepackage{pgfplots}
\usepackage{newlfont}
\usepackage[colorlinks,linkcolor=blue,citecolor=blue,
anchorcolor=blue,bookmarksopen,pdfpagetransition={Wipe}]{hyperref}
\usepackage{amsmath,amsthm,amssymb}
\usepackage{subfigure}
\usepackage{graphicx}
\usepackage{amsmath, amsfonts}
\usepackage{dsfont}
\usepackage{booktabs}
\usepackage{fourier}
\usepackage{array}
\usepackage{makecell}
\usepackage{algorithm}
\usepackage{algpseudocode}
\newtheorem{thm}{Theorem}[section]

\newtheorem{lem}[thm]{Lemma}

\newtheorem{defn}[thm]{Definition}

\newtheorem{example}[thm]{Example}
\numberwithin{equation}{section}
\DeclareMathOperator{\sign}{sign}

\begin{document}
\title{\textbf{Color image restoration with
impulse noise based on fractional-order total variation and framelet}}
\author{
R. Parvaz\footnote{Corresponding author, Email: reza.parvaz@yahoo.com, rparvaz@uma.ac.ir}}
\date{}
\maketitle
\begin{center}
Department of Mathematics, University of Mohaghegh Ardabili,
56199-11367 Ardabil, Iran.\\
\end{center}
\begin{abstract}
\noindent
Restore lost images due to noise and blurred is
a burgeoning subject in image processing and
despite the different algorithms on this subject,
but the effort to improve is always considered.
The definition of fractional derivatives in recent
years has created a powerful tool for this purpose.
In the present paper, using fractional-order total
variation and framelet transform,
the nonconvex model for
image restoration with impulse noise problem is improved.
Then by alternating direction method of multipliers (ADMM) and
primal-dual problem, the proposed model is solved.
 The convergence of the proposed
algorithm is studied.
And the proposed algorithm is evaluated using different types of tests.
The output results show the efficiency of proposed method.

\end{abstract}
\vskip.3cm \indent \textit{\textbf{Keywords:}}
Image deblurring; Fractional total variation; Framelet; Nonconvex.
\vskip.3cm

\section{Introduction}
The image blurring and noise process can be considered by using the following
formula
\begin{align*}
F=h\circledast U+N_\epsilon,
\end{align*}
where $F$ and $U \in \mathds{R}^{n \times m}$ are observed and original images, respectively, and
$N_\epsilon$ shows noise.
Also, $h\in \mathds{R}^{r \times s}$ is blurred kernel and $\circledast$ denotes
two-dimensional convolution operator. The above relation can be rewritten as a linear equation as follows
\begin{align*}
f=Au+n_\epsilon,
\end{align*}
where $f$ and $u\in \mathds{R}^{nm \times 1}$ denote the reshaped of observed and original images, respectively. Also $A \in \mathds{R}^{nm \times nm}$ is obtained according to the blurred kernel and boundary conditions for the blurred process. There are four main boundary conditions that include
zero, periodic, reflexive and anti-reflexive boundary conditions.
Depending on each of these boundary conditions, the matrix $A$ has a special structure.
For example in the zero boundary condition, $A$ is a block toeplitz
with toeplitz blocks (BTTB) matrix or in the periodic boundary condition,
$A$ is a block circulant with circulant blocks (BCCB) matrix.
Due to the noise and large size of the unknown coefficients  and ill-conditioned for this linear equation, it is not possible to solve this problem directly.
The first attempts to solve this problem can be made on SVD decompositions. But in most cases, this method does not have the desired output. Other method, that can be used for this problem is the basis pursuit
problem \cite{18} by minimization problem as
\begin{align*}
\min_{u} \{\|u \|_1: Au=f\},
\end{align*}
which is studied in many papers as \cite{19,20}. One of the most popular methods that has attracted the most attention is the use of total variation (TV) based on regularization scheme (the ROF model). This model is introduced by Rudin, Osher, and Fatemi \cite{21} as
\begin{align*}
\min_{u} \frac{\lambda}{2}\|Au-f\|^2_2+\int_{\Omega}|Du|,
\end{align*}
where $\Omega$ shows a bounded open subset with Lipschitzian boundary in $\mathds{R}^{2}$ and
$Du$ denotes the derivative of $u$. The idea of using total variation has been considered in recent years and has been used in various articles, for example \cite{22,23,24,25}.
One of the most important tools that has improved this method in recent years is the use of fractional derivatives. Despite the existence of different definitions for fractional derivatives, such as
Riemann-Liouville (R-L) and Caputo definition,
the Gr\"{u}nwald-Letnikov (G-L) definition is often used in image processing due to its lower computational complexity than other definitions. Based on these types of derivatives,
fractional-order total variation (FTV) is introduced and used in image and signal processing
\cite{26,27,28}. With the development of the wavelet and framelet concept, an efficient tool for
image and signal processing is developed that has been used in various articles as \cite{29,30,31,32}.
These concepts are introduced in the following subsections. The reader can find more information about these concepts in \cite{8,9,10}. The $l_1$ norm based on frame transform for restoration problem is introduced
by Dong, Ji and Shen \cite{36}. Due to the structure of $l_1$ norm as nonsmooth
and nonseparable, solving a problem that includes these norms is hard. The split Bregman algorithm
 is a suitable tool for solving these types of problems \cite{33,34,35}. In the proposed model
 a non-convex $(l_1-l_2)$-norm is considered as regularization term.
 $(l_1-l_2)$-norm is studied in many papers as \cite{3,7}. Also, this norm based on frame transform
is used by Jingjing Liu  et al in \cite{3}  as following model
\begin{align*}
\min _{u} \|Au-f\|_1+\lambda_1\big( \|Wu\|_1-\beta \|u\|_2\big),
\end{align*}
where $W$ represent the matrix of framelet transform.
As it is clear from this model, the total variation phrase is not observed in this model. Therefore, adding this expression can improve the algorithm.
Total variation method can preserve edges very well in the restored image,
then this helps to improve the restored image. Although the ordinary method of minimization is slow for this problem \cite{11,12},
the optimal algorithm based on primal-dual is introduced for solving
this problem in \cite{14,15,16,17}. \\

The organization of this paper is as follows: In Section 2,
the tools used in the proposed algorithm briefly introduced.
The details of the proposed method are given in Section 3.
The convergence analysis of the proposed algorithm is given in
Section 4. Simulation results and algorithm
analysis are studied in Section 5.
 Also the summary of paper is given in Section 6.

\section{Preliminaries}
\begin{defn}
Let $\mathcal{H}$ be a separable Hilbert space.
Then the sequence $F=\{f_i\}_{i\in I}\subseteq\mathcal{H}$
is named a frame in $\mathcal{H}$ if
there exist two constants as $\rho_1$ and $\rho_2$ such that for
all $f\in \mathcal{H}$
\begin{align*}
\rho_1 \|f\|^2 \leq \sum_{i\in I}|\langle f,f_i\rangle|^2 \leq \rho_2 \|f\|^2.
\end{align*}
When $\rho_1=\rho_2=1$, frame is called a Parseval frame in $\mathcal{H}=L_2(\mathds{R})$.
\end{defn}
In this paper, particular Parseval framelet systems in $\mathcal{H}$ that are constructed by
B-spline whose refinement mask is $h_0= \frac{1}{4}[1,2,1]$, with two corresponding framelet masks
$h_1=\frac{\sqrt{2}}{4}[1,0,-1]$ and $h_2=\frac{1}{4}[-1,2,-1]$.
Complete information about frame and framelet is available in \cite{8,9,10}, and the reader can refer to these sources.

\begin{defn}
The G-L fractional-order derivatives $D^{\alpha}_{x}$ and
$D^{\alpha}_{y}$ for input image $u$ of order $\alpha \in \mathds{R}^{+}$  are defined as
\begin{align*}
&D^{\alpha}_{x}u_{i,j}:=\sum_{l=0}^{k-1}\phi^\alpha_l u_{i-k,j},\\
&D^{\alpha}_{y}u_{i,j}:= \sum_{l=0}^{k-1}\phi^\alpha_l u_{i,j-k},
\end{align*}
where $\phi^\alpha_l=(-1)^l \frac{\Gamma(\alpha+1)}{\Gamma(l+1)\Gamma(\alpha+1-l)}$ and
$k$ denotes the number of neighboring pixels. Also the discrete fractional-order gradient
is considered as $\nabla^{\alpha}u=[D^{\alpha}_{x}u,D^{\alpha}_{y}u]^T$.
\end{defn}

\begin{defn}
The adjoint operators of the fractional-order derivatives are defined as
\begin{align*}
&(D^{\alpha}_{x})^Tu_{i,j}:=\sum_{l=0}^{k-1}\phi^\alpha_l u_{i+k,j},\\
&(D^{\alpha}_{y})^Tu_{i,j}:= \sum_{l=0}^{k-1}\phi^\alpha_l u_{i,j+k}.
\end{align*}
\end{defn}
Based on above definitions for a vector function $\textbf{p}(x,y)=\big(p_1(x,y), p_2(x, y)\big)$
the fractional divergence operator can be obtained as
\begin{align*}
div^{\alpha}\textbf{p}=(-1)^{\alpha}(\nabla^{\alpha})^T\textbf{p}=(-1)^{\alpha}\big( (D^{\alpha}_{x})^Tp^1_{i,j}+(D^{\alpha}_{y})^Tp^2_{i,j} \big).
\end{align*}

\section{Proposed minimization model and iterative method}
\subsection{Proposed model}
In the proposed model, the following model based on FTV is introduced for
restoration of image with impulse noise
\begin{align}\label{eq1}
\min _{u} \|Au-f\|_1+\lambda_1\big( \|Wu\|_1-\beta \|u\|_2\big)+\lambda_2\|u\|_{FTV},
\end{align}
where $\lambda_1,\lambda_2$ and $\beta$ are positive given parameters.
Also $\| \cdot\|_{FTV}$ denotes fractional-order total variation norm that defined as
\begin{align*}
\|u\|_{FTV}=\sum_{i,j}\sqrt{(D^{\alpha}_xu_{i,j})^2+(D^{\alpha}_yu_{i,j})^2}.
\end{align*}
This model contains a nonconvex sentence. There are several methods for solving nonconvergent optimization problems, for example see \cite{newr1}. In the proposed method, alternating direction method of multipliers (ADMM) is used for solving this problem. By auxiliary variables
$m=\{m_i\}^{3}_{i=1}$, \eqref{eq1} can be written as
\begin{align}\nonumber
 &  \min _{u} \|m_1\|_1+\lambda_1\big( \|m_2\|_1-\beta \|u\|_2\big)+\lambda_2\|m_3\|_{FTV}, \\ \label{eq2}
 &\text{s.t.}~~ m_1=Au-f,~~m_2=Wu,~~m_3=u.
\end{align}
The augmented Lagrangian for \eqref{eq2} is obtained as
\begin{align}\nonumber
 \min_{u,m,n}L(u,m_1,m_2,m_3&,n_1,n_2,n_3)=
 \min_{u,m,n}\|m_1\|_1+\lambda_1\big( \|m_2\|_1-\beta \|u\|_2\big)
 +\lambda_2\|m_3\|_{FTV} \\\nonumber
  &+\mu_1\langle Au-f-m_1,n_1\rangle
  +\frac{\mu_1}{2}\|Au-f-m_1\|^2_2+\mu_2\langle Wu-m_2,n_2\rangle\\\label{eq3}
  &+\frac{\mu_2}{2}\|Wu-m_2\|^2_2+\mu_3\langle u-m_3,n_3\rangle+\frac{\mu_3}{2}\|u-m_3\|^2_2,
\end{align}
where $n=\{n_i\}^{3}_{i=1}$ and $\mu_i>0,~i=1,2,3$ are the Lagrange multipliers and penalty parameters, respectively. The extended iterative algorithm for solving problem \eqref{eq3} based on ADMM is given as

\begin{align} \nonumber
u^{k+1}&=\arg \min_{u}-\lambda_1 \beta \|u\|_2+\mu_1\langle Au-f-m^k_1,n^k_1\rangle
+\frac{\mu_1}{2}\|Au-f-m^k_1\|^2_2\\\nonumber
&+\mu_2\langle Wu-m^k_2,n^k_2\rangle
+\frac{\mu_2}{2}\|Wu-m^k_2\|^2_2+\mu_3\langle u-m^k_3,n^k_3\rangle\\\label{eq4}
&+\frac{\mu_3}{2}\|u-m^k_3\|^2_2,\\ \label{eq5}
m^{k+1}_1&=\arg \min_{m_1}\|m_1\|_1+\mu_1\langle Au^{k+1}-f-m_1,n^k_1\rangle
  +\frac{\mu_1}{2}\|Au^{k+1}-f-m_1\|^2_2,
\end{align}

\begin{align} \label{eq6}
n^{k+1}_1&=n^k_1+Au^{k+1}-f-m^{k+1}_1,\\ \label{eq7}
m^{k+1}_2&=\arg \min_{m_2}\lambda_1 \|m_2\|_1+\mu_2\langle Wu^{k+1}-m_2,n^{k}_2\rangle+\frac{\mu_2}{2}\|Wu^{k+1}-m_2\|^2_2,\\\label{eq8}
n^{k+1}_2&=n^{k}_2+Wu^{k+1}-m^{k+1}_2,\\ \label{eq9}
m^{k+1}_3&=\arg \min_{m_3}\lambda_2 \|m_3\|_{FTV}
+\mu_3\langle u^{k+1}-m_3,n^{k}_3\rangle+\frac{\mu_3}{2}\|u^{k+1}-m_3\|^2_2,\\ \label{eq10}
n^{k+1}_3&=n^{k}_3+u^{k+1}-m^{k+1}_3.
\end{align}
\subsection{Solve subproblems}
In this subsection, the solution to each
 of the subproblems introduced in the previous subsection are studied.
 For \eqref{eq4} by using the optimal condition, we get
\begin{align}\nonumber
\big( \mu_1 A^T A&+\mu_3I+\mu_2I-\lambda_1 \beta/\| u^k\|_2I\big)u\\\label{eqn1}
&=\mu_1 A^T(f+m^k_1-n^k_1)+\mu_2 W^T(m^k_2-n^k_2)+\mu_3(m^k_3-n^k_3).
\end{align}
Under the periodic condition for blurred processing, the blurred matrix can be generated
as block circulant with circulant  blocks (BCCB) matrix.
It is well known in the field of linear algebra that these types of matrices are decomposed using Fourier transform. And due to the fast Fourier transform (FFT), the calculation is performed faster.
So by considering the periodic condition, Eq. \eqref{eqn1} is changed as
\begin{align}\nonumber
\big( \mu_1 \Lambda^{\ast} \Lambda&+(\mu_3+\mu_2-\lambda_1 \beta/\| u^k\|_2)I\big)Fu^{k+1}\\\label{eqn2}
&=\mu_1 \Lambda^{\ast} F(f+m^k_1-n^k_1)+\mu_2 F W^T(m^k_2-n^k_2)+\mu_3F(m^k_3-n^k_3),
\end{align}
where $\ast$  and $F$ denote the complex
conjugacy and fast Fourier transform, respectively.
Also $\Lambda$ is a diagonal matrix dependent on the
eigenvalues of the blurred matrix.
The equation \eqref{eqn2} can be easily solved using Fourier calculations.\\

Subproblem \eqref{eq5} can be rewritten as
\begin{align*}
m^{k+1}_1&=\arg \min_{m_1}\|m_1\|_1+\frac{\mu_1}{2}\|m_1-(Au^{k+1}-f+n^k_1)\|^2_2.
\end{align*}
The solution for the above problem based on proximal mapping for $l_1$-norm can be obtained
as
\begin{align}\label{eq12}
m^{k+1}_1=\Psi_{1/\mu_1}(Au^{k+1}-f+n^k_1),
\end{align}
where $\Psi$ is defined as
\begin{align*}
\Psi_{a}(x)=\sign(x)\max(|x|-a,0).
\end{align*}
In a similar way, \eqref{eq7} can be written as
\begin{align*}
m^{k+1}_2&=\arg \min_{m_2}\lambda_1 \|m_2\|_1
+\frac{\mu_2}{2}\|m_2-(Wu^{k+1}+n^{k}_2)\|^2_2,
\end{align*}
 and the solution by proximal mapping is earned as
\begin{align}\label{eq14}
m^{k+1}_2=\Psi_{\lambda_1/\mu_2}(Wu^{k+1}+n^{k}_2).
\end{align}
In order to avoid complex calculations for Eq. \eqref{eq9}, the dual problem is used
to find the solution. In the first step, Eq. \eqref{eq9} is changed as
\begin{align}\label{eq15}
  m^{k+1}_3&=\arg \min_{m_3}\lambda_2 \|m_3\|_{FTV}
+\frac{\mu_3}{2}\|m_3-(u^{k+1}-n^{k}_3)\|^2_2,
\end{align}
then in the next step, the dual problem is obtained and solved for this problem.
\begin{lem}\label{lem1}
If we consider $J(u)=\min_{u}\|u\|_{FTV}$ then the dual problem of $J(u)$ is obtained as \cite{16}
\begin{align*}
J(u)=\sup_{\textbf{p}}~\langle \textbf{p}, \nabla^{\alpha}u \rangle-J^{\ast}(\textbf{p}),
\end{align*}
where
\begin{align*}
J^{\ast}(\textbf{p})= \left\{ \begin{array}{l}
 \begin{array}{*{20}{c}}
   0, & \text{if}~~|\textbf{p}|\leq 1,   \\
\end{array} \\
 \begin{array}{*{20}{c}}
   \infty, & \text{if}~~|\textbf{p}|> 1.   \\
\end{array} \\
 \end{array} \right.
\end{align*}
\end{lem}
Based on Lemma \ref{lem1}, the corresponding primal-dual problem of
\eqref{eq15} is written as follows

\begin{align}\label{eq16}
(m^{k+1}_3,\textbf{p}^{k+1})=\arg\min_{m_3}\max_{\textbf{p}\in \chi}
~ \frac{\lambda_2}{\mu_3}\langle \textbf{p}, \nabla^{\alpha}m_3 \rangle
+\frac{1}{2}\|m_3-(u^{k+1}-n^{k}_3)\|^2_2,
\end{align}
where $\chi=\big \{\textbf{p}\in \mathds{R}^{2nm}
|\textbf{p}_i\in \mathds{R}^2, \|\textbf{p}_i\|_2\leq 1, \forall i\in \{1,\ldots,nm\} \big \}$.
By the iterative scheme, the solution of the primal-dual problem \eqref{eq16}
can be written as
\begin{align}\label{eq17}
&\textbf{p}^{k+1}=\arg\max_{\textbf{p}\in \chi}
\frac{\lambda_2}{\mu_3}\langle \textbf{p}, \nabla^{\alpha}\hat{m}_3^{k}\rangle
-\frac{1}{2\gamma}\|\textbf{p}-\textbf{p}^k\|^2_2,
\\\label{eq18}
&m_3^{k+1}=\arg \min_{m_3}
~ \frac{\lambda_2}{\mu_3}\langle \textbf{p}^{k+1}, \nabla^{\alpha}m_3 \rangle
+\frac{1}{2}\|m_3-(u^{k+1}-n^{k}_3)\|^2_2,
\\\label{nn3}
&\hat{m}_3^{k+1}=2m^{k+1}_3-m^k_3.
\end{align}
After simplifying \eqref{eq17}-\eqref{eq18},
the following statements are obtained
\begin{align}\label{nn1}
&\textbf{p}^{k+1}=\frac{\textbf{p}^{k}+\frac{\lambda_2 \gamma}{\mu_3}\nabla^{\alpha} \hat{m}_3^{k} }
{\max(|\textbf{p}^{k}+\frac{\lambda_2 \gamma}{\mu_3}\nabla^{\alpha} \hat{m}_3^{k} |,1)},\\\label{nn2}
&m^{k+1}_3=m^k_3-\tau\big(
\frac{\lambda_2 \gamma}{\mu_3}(\nabla^{\alpha})^{T}\textbf{p}^{k+1}
+m^k_3-(u^{k+1}+n^k_3)\big),
\end{align}
where $\tau$ is step size and $\gamma$ is a positive constant.\\
The general structure of the proposed method for
restoration of blurred images with impulse noise
is summarized in Algorithm \textcolor{blue}{1}.

\noindent\line(1,0){280}\\
\vspace{-0.3cm}
Algorithm 1: Restoration proposed algorithm.\\
\vspace{-0.3cm}
\noindent\line(1,0){280}\\
\begin{algorithmic}
\State   \textbf{Initialization:} $\alpha,\beta,\tau,\gamma,
\{\lambda_i\}^2_{i=1}, \{\mu_i\}^3_{i=1}, m^0, n^0,u^0$.
\For{k=1,\ldots}
\State Compute $u^{k}$ by solving \eqref{eqn2},
\State Compute $m^{k}_1$ by solving \eqref{eq12},
\State Compute $n^{k}_1$ by solving \eqref{eq6},
\State Compute $m^{k}_2$ by solving \eqref{eq14},
\State Compute $n^{k}_2$ by solving \eqref{eq8},
\State Compute $m^{k}_3$ by solving \eqref{nn3}-\eqref{nn2},
\State Compute $n^{k}_3$ by solving \eqref{eq10},
\EndFor
\State  If the stop condition is met in the above step, stop the loop.
\end{algorithmic}
\line(1,0){280}\\
\subsection{Extended algorithm for color image}
The proposed restoration algorithm for grayscale image is
 introduced in the previous subsection.
But this algorithm can be extended to color images.
In the following, the details of the proposed algorithm
for color images
are explained. The blurred matrix for color image
is given in $\mathds{R}^{3nm \times 3nm}$ as
\begin{align*}
A = \left[ {\begin{array}{*{20}{c}}
   {\begin{array}{*{20}{c}}
   A_{rr}  \\
   A_{gr}  \\
   A_{br}  \\
\end{array}} & {\begin{array}{*{20}{c}}
   A_{rg}  \\
   A_{gg}  \\
   A_{bg}  \\
\end{array}} & {\begin{array}{*{20}{c}}
   A_{rb}  \\
   A_{gb}  \\
   A_{bb}  \\
\end{array}}  \\
\end{array}} \right].
\end{align*}
Also, in the calculations related to this subsection, for different layers of the color image, we expand the values that is calculated in the previous subsection as
$u=[u_r;u_g;u_b]$, $f=[f_r;f_g;f_b]$, $m_i=[m_{i,r};m_{i,g};m_{i,b}]$, $n_i=[n_{i,r};n_{i,g};n_{i,b}]$
for $i=1,2,3$.\\
Using the same procedure mentioned in the pervious section the following phrase is obtained
\begin{align}\nonumber
\big(\mu_1 \Lambda^{\ast}\Lambda&+(\mu_3+\mu_2-\lambda \beta/\|u\|_2)I\big)\hat{F}u^{k+1}\\ \label{nm1}
&=\mu_1 \Lambda^{\ast}\hat{F}(f+m^k_1-n^k_1)+\mu_2\hat{F}W^T(m^k_2-n^k_2)
+\mu_3\hat{F}(m^k_3-n^k_3),
\end{align}
where $\hat{F}=I\otimes F$, here $\otimes$ shows the Kronecker product. Also $\Lambda$
is obtained after decomposed using Fourier transform as
\begin{align*}
\Lambda = \left[ {\begin{array}{*{20}{c}}
   {\begin{array}{*{20}{c}}
   \Lambda_{rr}  \\
   \Lambda_{gr}  \\
   \Lambda_{br}  \\
\end{array}} & {\begin{array}{*{20}{c}}
   \Lambda_{rg}  \\
   \Lambda_{gg}  \\
   \Lambda_{bg}  \\
\end{array}} & {\begin{array}{*{20}{c}}
   \Lambda_{rb}  \\
   \Lambda_{gb}  \\
   \Lambda_{bb}  \\
\end{array}}  \\
\end{array}} \right].
\end{align*}
Other parts of the proposed algorithm are
similar to those described in the pervious
section, and the calculation details are given in Algorithm
\textcolor{blue}{2}. The simulation results of the algorithms
described in this section are studied in the simulation results section.

\noindent\line(1,0){280}\\
\vspace{-0.3cm}
Algorithm 2: Color image restoration proposed algorithm.\\
\vspace{-0.3cm}
\noindent\line(1,0){280}\\
\begin{algorithmic}
\State   \textbf{Initialization:} $\alpha,\beta,\tau,\gamma,
\{\lambda_i\}^2_{i=1}, \{\mu_i\}^3_{i=1}, m^0, n^0,u^0$.
\For{k=1,\ldots}
\State Compute $u^{k}$ by solving \eqref{nm1},
\For{l=r,g,b}
\State Compute $m^{k}_{1,l}$ by solving: $m^{k}_{1,l}=\Psi_{1/\mu_1}(Au^{k}_l-f_l+n^{k-1}_{1,l})$,
\State Compute $n^{k}_{1,l}$ by solving: $n^{k}_{1,l}=n^{k-1}_{1,l}+Au^{k}_l-f-m^{k}_{1,l}$,
\State Compute $m^{k}_{2,l}$ by solving: $m^{k}_{2,l}=\Psi_{\lambda_1/\mu_2}(Wu^{k}_l+n^{k-1}_{2,l})$,
\State Compute $n^{k}_{2,l}$ by solving: $n^{k}_{2,l}=n^{k-1}_{2,l}+Wu^{k}_l-m^{k}_{2,l}$,
\State Compute $m^{k}_{3,l}$ by solving:
\State $\textbf{p}^{k}_l=\frac{\textbf{p}^{k-1}_l+\frac{\lambda_2 \gamma}{\mu_3}\nabla^{\alpha} \hat{m}_{3,l}^{k-1} }
{\max(|\textbf{p}^{k-1}_l+\frac{\lambda_2 \gamma}{\mu_3}\nabla^{\alpha} \hat{m}_{3,l}^{k-1} |,1)}$,
\State $m^{k}_{3,l}=m^{k-1}_{3,l}-\tau\big(
\frac{\lambda_2 \gamma}{\mu_3}(\nabla^{\alpha})^{T}\textbf{p}^{k}_l
+m^{k-1}_{3,l}-(u^{k}_l+n^{k-1}_{3,l})\big)$,
\State $\hat{m}_{3,l}^{k}=2m^{k}_{3,l}-m^{k-1}_{3,l}$,
\State Compute $n^{k}_{3,l}$ by solving: $n^{k}_{3,l}=n^{k-1}_{3,l}+u^{k}_l-m^{k}_{3,l}$,
\EndFor
\EndFor
\State  If the stop condition is met in the above step, stop the loop.
\end{algorithmic}
\line(1,0){280}\\

\section{Convergence analysis}
In this section, the convergence of the proposed method is studied. The convergence analysis for proposed method is based on the described method in \cite{3,7}.
\begin{lem}\label{lm2}
Let the objective function be increasing, that is
$\|Au-f\|_1+\lambda_1\big( \|Wu\|_1-\beta \|u\|_2\big)+\lambda_2\|u\|_{FTV}\rightarrow  \infty$
when $\| x\|_2 \rightarrow  \infty$. Also assume that
$\{u^k,m_1^k,m_2^k,m_3^k,n_1^k,n_2^k,n_3^k\}$ be the sequence generated by the proposed method,
then the following statements hold:\\

a)
\begin{align*}
L(u^{k+1},m_1^{k+1},m_2^{k+1},&m_3^{k+1},n_1^{k+1},n_2^{k+1},n_3^{k+1})
-L(u^{k},m_1^{k},m_2^{k},m_3^{k},n_1^{k},n_2^{k},n_3^{k})\\
&\leq \sum^{3}_{i=1}C_i\|m^{k+1}_i-m^{k}_i\|^2_2.
\end{align*}

b) If there exists a $p \in \partial_uL(u^{k+1},m_1^{k+1},m_2^{k+1},m_3^{k+1},n_1^{k+1},n_2^{k+1},n_3^{k+1})$
then
\begin{align*}
\| p\|_2+\sum^3_{i}(\partial_{m_i}&+\partial_{n_i})
L(u^{k+1},m_1^{k+1},m_2^{k+1},m_3^{k+1},n_1^{k+1},n_2^{k+1},n_3^{k+1})\\
& \leq \sum^{3}_{i=1}C_{i+3} \|m^{k+1}_i-m^k_i\|_2,
\end{align*}

where in conditions (a) and (b), $\{C_i\}^{6}_{i=1}$ are constant.

\end{lem}
\begin{proof}
a) By problem \eqref{eq4}, we get
\begin{align*}
L(u^{k+1},m_1^{k},m_2^{k},m_3^{k},n_1^{k},n_2^{k},n_3^{k})
-L(u^{k},m_1^{k},m_2^{k},m_3^{k},n_1^{k},n_2^{k},n_3^{k})
\leq 0.
\end{align*}
Also by using \eqref{eq5}-\eqref{eq10}, we obtain
\begin{align*}
L(u^{k+1},&m_1^{k+1},m_2^{k+1},m_3^{k+1},n_1^{k+1},n_2^{k+1},n_3^{k+1})
-L(u^{k},m_1^{k},m_2^{k},m_3^{k},n_1^{k},n_2^{k},n_3^{k})\\
&\leq \|m^{k+1}_1\|_1-\|m^{k}_1\|_1+\lambda_1\big(\|m^{k+1}_2\|_1-\|m^{k}_2\|_1\big)
+\lambda_2\big(\|m^{k+1}_3\|_{FTV}-\|m^{k}_3\|_{FTV}\big)\\
&+\sum^{3}_{i=3}\mu_i \big(\|n^{k+1}_i-n^{k}_i\|^2_2-\langle n^{k}_i,m^{k+1}_i-m^{k}_i\rangle
-\frac{1}{2}\|m^{k+1}_i-m^k_i\|^2_2\\
&-\langle m^{k+1}_i-m^k_i, n^{k+1}_i-n^k_i\rangle\big).
\end{align*}
Based on optimality conditions for \eqref{eq5}-\eqref{eq10}, we have
$\partial_{m_1}\|m^{k+1}\|_1=\mu_1 n^{k+1}_1$,~
$\lambda_1 \partial_{m_2}\|m_2^{k+1}\|_1=\mu_2n^{k+1}_2$
and
$\lambda_2 \partial_{m_3}\|m_3^{k+1}\|_{FTV}=\mu_3n^{k+1}_3$,
then by using Lipschitz continuous gradient and Young's inequality, the following inequalities
for any positive $\{c_i\}^3_{i=1}$ hold
\begin{align*}
& \| m^{k+1}_1\|_1-\| m^{k}_1\|_1-\mu_1 \langle n^{k}_1, m^{k+1}_1-m^{k}_1 \rangle
\leq \frac{L_1}{2}\|m^{k+1}_1-m^{k}_1 \|^2_2,\\
& \lambda_1\big(\| m^{k+1}_2\|_1-\| m^{k}_2\|_1\big)
-\mu_2 \langle n^{k}_2, m^{k+1}_2-m^{k}_2 \rangle
\leq \frac{L_2}{2}\|m^{k+1}_2-m^{k}_2 \|^2_2,\\
& \lambda_2\big(\| m^{k+1}_3\|_{FTV}-\| m^{k}_3\|_{FTV}\big)
-\mu_3 \langle n^{k}_3, m^{k+1}_3-m^{k}_3 \rangle
\leq \frac{L_3}{2}\|m^{k+1}_3-m^{k}_3 \|^2_2,\\
&-\mu_i \langle m^{k+1}_i-m^{k}_i,n^{k+1}_i-n^{k}_i \rangle
\leq c_i\mu_i\|n^{k+1}_3-n^k_3\|^2_2+\frac{\mu_i}{4c_i}\|m^{k+1}_3-m^{k}_3\|^2_2,
~~i=1,2,3,
\end{align*}
if $\{c_i\}^3_{i=1}$ is chosen as $\{\mu_i/2L_i\}^3_{i=1}$, we get
\begin{align*}
L(u^{k+1},&m_1^{k+1},m_2^{k+1},m_3^{k+1},n_1^{k+1},n_2^{k+1},n_3^{k+1})
-L(u^{k},m_1^{k},m_2^{k},m_3^{k},n_1^{k},n_2^{k},n_3^{k})\\
&\leq \sum^{3}_i \big(\frac{L_i}{2}-\frac{\mu_i}{2}+\frac{\mu_1}{4c_i} \big)
\|m^{k+1}_i-m^k_i\|^2_2+(1+c_i)\mu_i\|n^{k+1}_i-n^k_i\|^2_2\\
&=\sum^{3}_{i=1}C_i\|m^{k+1}_i-m^k_i\|^2_2,
\end{align*}
where $C_i=(3L_i-\mu_i)/2+L_i^2/\mu_i$.
Therefore, the first part of the lemma is proved.\\
b) By using optimality condition for
\eqref{eq4}, there is a $q\in \partial_{u}\big(-\lambda_1 \beta_1 \|u\|_2\big)$
such that
\begin{align*}
q&+\mu_1A^T \big(Au^{k+1}-f+n^k_1-m^k_1\big)+\mu_2 W^T\big(Wu^{k+1}+n^k_2-m^{k}_2 \big)\\
&+\mu_3 \big(u^{k+1}+n^k_3-m^k_3\big)=0,
\end{align*}
now let
\begin{align*}
p&=q+\mu_1A^T \big(Au^{k+1}-f+n^{k+1}_1-m^{k+1}_1\big)+\mu_2 W^T\big(Wu^{k+1}+n^{k+1}_2-m^{k+1}_2 \big)\\
&+\mu_3 \big(u^{k+1}+n^{k+1}_3-m^{k+1}_3\big)
\in \partial_{u}L(u^{k+1},m^{k+1}_1,m^{k+1}_2,m^{k+1}_3,n^{k+1}_1,n^{k+1}_2,n^{k+1}_3),
\end{align*}
then we obtain
\begin{align}\label{nn77}
\|p\|_2 \leq \sum^{3}_{i=1}c_{i+3}(\mu_i+L_i)\|m^{k+1}_i-m^{k}_i \|_2,
\end{align}
where $\{c_i\}^6_{i=4}$ are constant. From the optimality condition of
\eqref{eq4}-\eqref{eq10}, the following inequalities for $i=1,2,3$ are obtained
\begin{align*}
&\|\partial_{m_i}L(u^{k+1},m^{k+1}_1,m^{k+1}_2,m^{k+1}_3,n^{k+1}_1,n^{k+1}_2,n^{k+1}_3) \|_2
\leq L_i \|m^{k+1}_i-m^{k}_i\|_2, \\
&\|\partial_{n_i}L(u^{k+1},m^{k+1}_1,m^{k+1}_2,m^{k+1}_3,n^{k+1}_1,n^{k+1}_2,n^{k+1}_3) \|_2
\leq L_i \|m^{k+1}_i-m^{k}_i\|_2.
\end{align*}
Therefore, by combining the above inequalities with \eqref{nn77}, the relation (b) is obtained.
\end{proof}
Using the above lemma, the final theorem for convergence is considered as follows. Due to the similarity of the proof with \cite{7}, the proof is omitted.
\begin{thm}
Considering the assumptions of Lemma \ref{lm2} and let
$\mu_i>(3+\sqrt{17})L_i/2$ for $i=1,2,3$, then\\
a) the sequence $\big( u^{k},\{m^{k}_i\}^3_{i=1},\{n^{k}_i\}^3_{i=1} \big)$ generated by proposed method is bounded and has at least one limit point.\\
b) $\|u^{k+1}-u^k\|_2 \rightarrow 0$, $\|m^{k+1}_i-m^k_i\|_2 \rightarrow 0$ and $\|n_i^{k+1}-n_i^k\|_2 \rightarrow 0$ for $i=1,2,3$.\\
c) each limit point $(u^{\ast},m_1^{\ast},m_2^{\ast},m_3^{\ast},n_1^{\ast},n_2^{\ast},n_3^{\ast})$ is a stationary point of $L(u,m_1,m_2,m_3,n_1,n_2,n_3)$, and $u^{\ast}$ is a stationary
point of the proposed model \eqref{eq1}.
\end{thm}
\section{Simulation results}
The simulation results of the proposed algorithm that introduced in
the previous sections are studied in this section. For simulation results, a computer based on Windows 10-64bit, Intel(R) Core(TM) i3-5005U CPU @2.00GHz, by matlab 2014b is used.
Also stopping criterion in this section is considered as
\begin{align*}
\frac{\|u^{k+1}-u^{k}\|_2}{\|u^{k+1} \|_2}\leq tol,
\end{align*}
where $u^k$ is the restored image at the $k$th iteration.
In all examples, the $tol$ equal to $10^{-3}$ is selected.
In the analysis of the proposed algorithm various tests including
 peak signal to noise ratio (PSNR), signal-to-noise
ratio (SNR), structural similarity (SSIM), Feature Similarity (FSIM)
and Relative Error (ReErr) are studied. More information on these
criteria value can be found in \cite{4,5,6}. Also ReErr is calculated as
\begin{align*}
ReErr=\frac{\| u^k-u\|^2_2}{\|u\|^2_2}.
\end{align*}

In the following examples, $M(len,\theta)$ denotes the
linear motion blur of a camera by $len$ pixels,
with an angle of $\theta$ degrees in a counterclockwise direction,
$A([r_1,r_2])$ shows the average blur with a size $[r_1,r_2]$,
$G(hsize, \sigma)$ denotes the Gaussian blur
of size $hsize$ with standard deviation $\sigma$ (positive).

\begin{example}
\normalfont
(Grayscale image) In this example, Algorithm \textcolor[rgb]{0,0,1}{1} is used for different images
from the USC-SIPI images database\footnote{http://sipi.usc.edu/database/}, i.e.,
``5.1.12 ($256 \times 256 $)'', ``5.2.10 ($512 \times 512$)'' and ``7.1.09 ($512 \times 512$)''.
Also ``4.2.05'' is changed to image with size ($256 \times 256 $) by Matlab internal functions and used in the simulation. For images ``5.1.12'' and ``5.2.10'', $A([11,11])$ with ``salt $\&$ pepper'' and ``Random-valued'' noise, respectively, with
different density are used, for image ``7.1.09'', blur kernel $G(7,4)$
with ``salt $\&$ pepper'' noise with
different densities are used.
Table \ref{tab1} compares numerical results for different densities value by proposed algorithm and method in
\cite{3}. These results indicate an improvement in the output of the proposed algorithm.
Also the blurred and noisy images for ``5.2.10'' are shown in Figure \ref{fig1},
and the restored images corresponding to each image are given below each image
in this figure. Also in order to show the effect of blurred kernel and noise on the image,
in Figure \ref{fig2} (a-d), each of the steps of blurring by $M(35,135)$ and adding noise by ``random-valued'' noise with
$10 \%$ density and restoring image is given separately.
In this case we have $PSNR=33.031,~SNR=30.241$, $SSIM= 0.92600$ and $FSIM=0.99712$.
The results in this example show a good representation of the proposed algorithm.
\end{example}
\begin{center}
\begin{figure}[h!]
\centering
\includegraphics[width=1.1\textwidth]{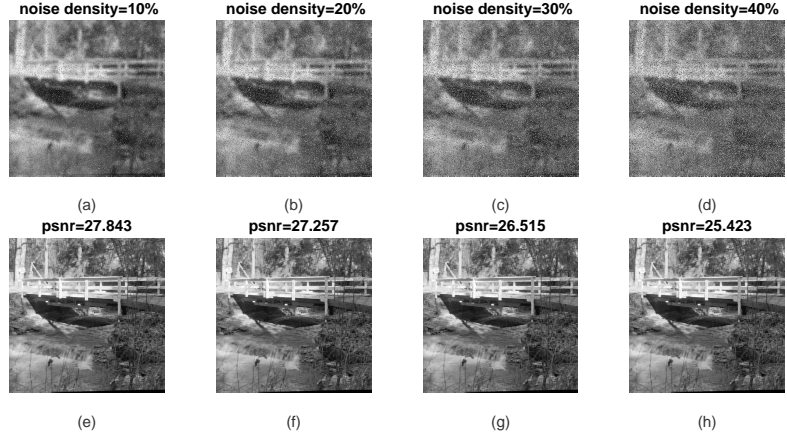}
\vspace{-1cm}
\caption{{\footnotesize
(a-d) Blurred and noisy images, (e-h) restored images for 5.2.10.
} }
\label{fig1}
\end{figure}
\end{center}

\begin{center}
\begin{table}[ht!]
\footnotesize
\caption{{\footnotesize Test results for different grayscale images.}}
\label{tab1}
\centering
\begin{tabular}{ccccccccccccccc}
\hline
imgae   & density  &method   &PSNR    &SNR     &ISNR   &ReErr       &SSIM    &FSSIM\\
\hline
        & $10\%$  &Proposed  &35.741  &33.393  &21.702 &0.00045781  &0.95021 &0.9991\\
5.1.12  &         &\cite{3}       & 33.682 &31.334  &19.55  &0.00073556  &0.95406 &0.9988\\
        & $30\%$  &Proposed  &28.344  &25.995  &18.746 &0.0025146   &0.89537 &0.99418\\
        &         &\cite{3}      &27.769  &25.421  &13.78  &0.00287     &0.90592 &0.99223\\
\hline
        & $10\%$  &Proposed  &27.843  &21.739  &11.106  &0.0067000  &0.83063 &0.99917\\
5.2.10  &         &\cite{3}       &27.419  &21.315  &10.654  &0.0073876  &0.81419 &0.99900\\
        & $30\%$  &Proposed  &26.515  &20.411  &13.120   &0.0090975  &0.77559 &0.99851\\
        &         &\cite{3}       &26.327  &20.223  &12.929  &0.0094992  &0.76804 &0.99835\\
\hline
        & $10\%$  &Proposed  &33.699  &27.89   &18.254  &0.0016255  &0.89393  &0.99977\\
7.1.09  &         &\cite{3}       &33.254  &27.445  &17.903  &0.001801   &0.88131  &0.99976\\
        & $30\%$  &Proposed  &32.385  &26.576  &21.557  &0.0021999  &0.86585  &0.99963\\
        &         &\cite{3}      &31.707  &25.898  &20.87   &0.0025715  &0.84426  &0.99961\\
\hline
\end{tabular}
\end{table}
\end{center}

\begin{example}
\normalfont
(Color image)
In this example, Algorithm \textcolor[rgb]{0,0,1}{2} is studied for different color
images. In the simulation results, ``Lena'' $(512 \times 512 )$ image is blurred by following kernel
\begin{align*}
A = \left[ {\begin{array}{*{20}{c}}
   {\begin{array}{*{20}{c}}
   0.7A([15,15])  \\
   0.1G(21,11)  \\
   0.0M(41,90)\\
\end{array}} & {\begin{array}{*{20}{c}}
   0.15G(11,9)  \\
   0.8A([17,17])  \\
     0.2M(21,45)\\
\end{array}} & {\begin{array}{*{20}{c}}
   0.15G(31,13)  \\
   0.1A([13,13]) \\
   0.6M(61,135)  \\
\end{array}}  \\
\end{array}} \right],
\end{align*}
``House'' $(512 \times 512 )$ image is blurred by
\begin{align*}
  A=\frac{1}{3} diag\big(A([5,5]),A([7,7]),A([9,9])\big),
\end{align*}
blurred kernel for ``Peppers'' image $(256 \times 256 )$ is considered as
\begin{align*}
A = \left[ {\begin{array}{*{20}{c}}
   {\begin{array}{*{20}{c}}
   0.8A([11,11])  \\
   0.15A([11,11])  \\
   0.2A([11,11])\\
\end{array}} & {\begin{array}{*{20}{c}}
   0.1G(11,5)  \\
   0.7G(11,5)   \\
   0.2G(11,5) \\
\end{array}} & {\begin{array}{*{20}{c}}
   0.1M(21,135)  \\
   0.15M(21,135) \\
   0.6M(21,135) \\
\end{array}}  \\
\end{array}} \right],
\end{align*}
and finally blurred kernel $A=B\otimes M(41, 135)$ is used for ``Plate'' with size ($256 \times 256$), where
\begin{align*}
A = \left[ {\begin{array}{*{20}{c}}
   {\begin{array}{*{20}{c}}
   0.8  \\
   0.15  \\
   0.2\\
\end{array}} & {\begin{array}{*{20}{c}}
   0.1  \\
   0.7   \\
   0.2\\
\end{array}} & {\begin{array}{*{20}{c}}
   0.1  \\
   0.15\\
   0.6\\
\end{array}}  \\
\end{array}} \right].
\end{align*}
Also for the noise process,
``salt $\&$ pepper'' noise is used for ``Lena'', ``House'' and ``Plate'' images and
``random-valued'' noise is used for ``Peppers'' image.
Table \ref{tab2} compares the results of the proposed method
 with the results of the algorithms in \cite{1}
and \cite{2}.
Other results are also given in Table \ref{tab3} for test images.
The results show the efficiency of the proposed
algorithm in comparison with these algorithms.
Also, the proposed algorithm outputs for different images with different noise densities
are given in Figures \ref{fig3}-\ref{fig5}.
Figure \ref{fig2} (e-h) shows the effect of the blurred kernel and noise on the color image.
In this results ``salt $\&$ pepper'' noise with $30\%$ density is used.
Test results for this case are included
$PSNR=28.871$, $SNR=27.049$, $ISNR=20.636$, $ReErr=0.0019729$,
$SSIM=0.9535$ and $FSIM=0.99884$.

\end{example}

\begin{center}
\begin{figure}[h!]
\centering
\includegraphics[width=1.1\textwidth]{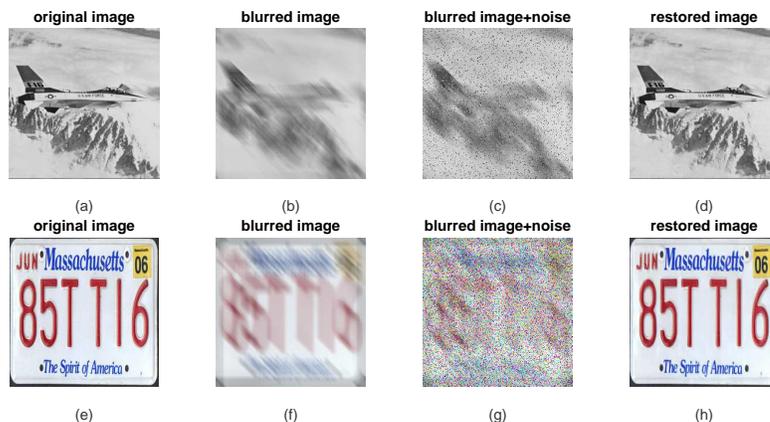}
\vspace{-1cm}
\caption{{\footnotesize
(a,e) Original images, (b,f) blurred images, (c,g) blurred and noisy images, (d,h) restored images.
} }
\label{fig2}
\end{figure}
\end{center}

\begin{center}
\begin{table}[ht!]
\footnotesize
\caption{{\footnotesize Test results for different color images.}}
\label{tab2}
\centering
\begin{tabular}{ccccccccccccccc}
\hline
&density&\multicolumn{2}{c}{$10\%$}   &\multicolumn{2}{c}{$20\%$  }   &\multicolumn{2}{c}{$30\%$  } &\multicolumn{2}{c}{$40\%$  }\\\cmidrule(l){3-4} \cmidrule(l){5-6} \cmidrule(l){7-8} \cmidrule(l){9-10}
imgae&method    &SNR &SSIM&SNR &SSIM &SNR &SSIM &SNR &SSIM\\
\hline
     &Proposed  &27.23  &0.9620    &26.83 &0.9591 &26.04&0.9524&24.04&0.9348\\
House&\cite{1}  &22.70  &0.9103    &22.33  &0.9045 &22.08&0.9027&21.61&0.8904\\
     &\cite{2}   &25.54  &0.9185    &24.57  &0.9071 &23.36&0.9024&22.58&0.8876\\
\hline
     &Proposed  &25.15  &0.9744    &24.96  &0.9734 &24.66&0.9717&24.33&0.9699\\
Lena &\cite{1}  &23.50  &0.8831    &23.48  &0.8841 &23.87&0.8907&23.71&0.8864\\
     & \cite{2}   &25.12  &0.9035    &24.48  &0.8941 &23.88&0.8860&23.90&0.8856\\
\hline
        &Proposed  &26.34  &0.9895    &25.80 &0.9877 &24.49&0.9826&20.43&0.9547\\
Peppers &\cite{1}  &22.34  &0.8530    &21.79  &0.8379 &20.88&0.8208&19.88&0.7941\\
         &\cite{2}    &23.85  &0.8681    &22.88  &0.8504 &21.73&0.8302&20.39&0.7982\\
\hline
\end{tabular}
\end{table}
\end{center}

\begin{center}
\begin{figure}[h!]
\centering
\includegraphics[width=1.1\textwidth]{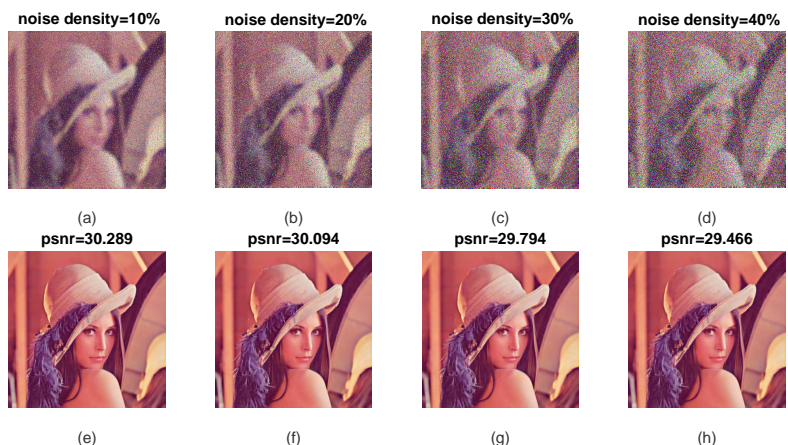}
\vspace{-1cm}
\caption{{\footnotesize
(a-d) Blurred and noisy images, (e-h) restored images for Lena $(512 \times 512)$.
} }
\label{fig3}
\end{figure}
\end{center}

\begin{center}
\begin{figure}[h!]
\centering
\includegraphics[width=1.1\textwidth]{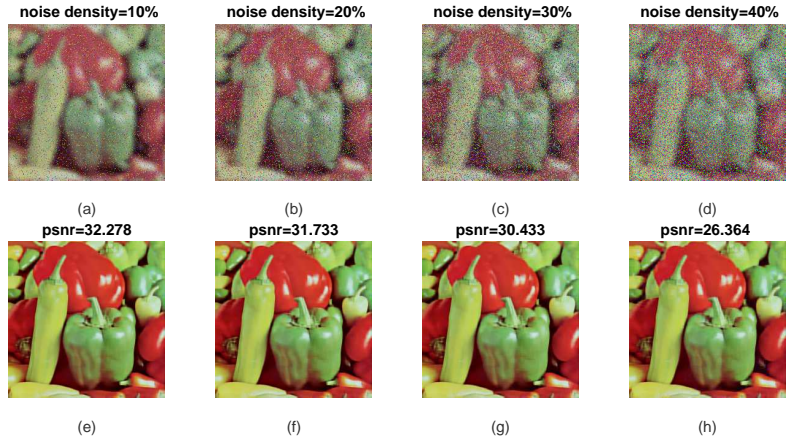}
\vspace{-1cm}
\caption{{\footnotesize
(a-d) Blurred and noisy images, (e-h) restored images for Peppers $(256\times 256)$.
} }
\label{fig4}
\end{figure}
\end{center}

\begin{center}
\begin{figure}[h!]
\centering
\includegraphics[width=1.1\textwidth]{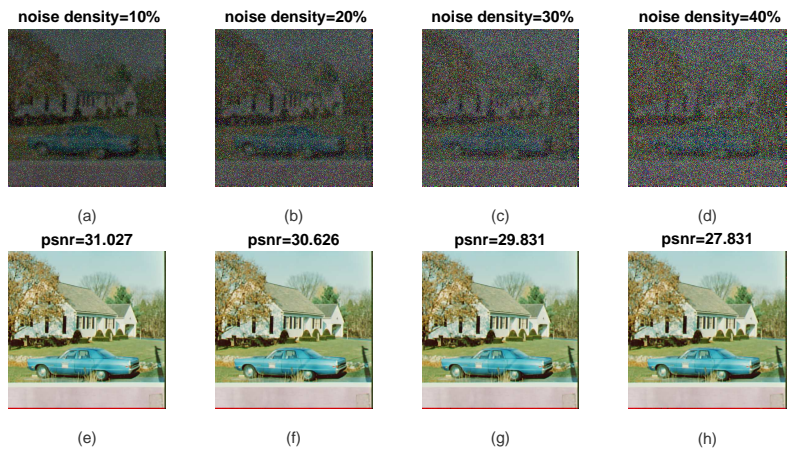}
\vspace{-1cm}
\caption{{\footnotesize
(a-d) Blurred and noisy images, (e-h) restored images for Househ
 $(512 \times 512)$.
} }
\label{fig5}
\end{figure}
\end{center}

\begin{center}
\begin{table}[ht!]
\footnotesize
\caption{{\footnotesize Test results for different grayscale images with different noise densities.}}
\label{tab3}
\centering
\begin{tabular}{ccccccccccccccc}
\hline
imgae   & density            &ISNR   &ReErr           &FSIM\\
\hline
        & $10\%$         &24.054    &0.0018923    &0.99976\\
House  & $20\%$         &23.907    &0.0020754    &0.99973\\
        & $30\%$        &23.342    &0.0024927      & 0.99965\\
        & $40\%$         &21.567    &0.0039502    &0.99947\\
\hline
        & $10\%$       &16.234   &0.0030536  &0.99914\\
Lena    & $20\%$       &18.433   &0.0031943  &0.99908\\
        & $30\%$        &19.703   &0.0034223  &0.99898\\
        & $40\%$        &20.507   &0.0036912  &0.99886\\
\hline
        & $10\%$       &16.785   &0.0023246   &0.99724\\
Peppers & $20\%$       &17.983   &0.0026354   &0.99716\\
        & $30\%$      &17.946   &0.0035549   &0.99689\\
        & $40\%$      &14.841   &0.0090734   &0.99597\\
\hline
\end{tabular}
\end{table}
\end{center}

\section{Conclusion}
In this paper the nonconvex model based on
fractional-order total variation and framelet transfer
is introduced for image restoration. The proposed model
is solved by ADMM and primal-dual methods. In the analysis
of the proposed algorithm
convergence analysis is studied and simulation
results are evaluated.
The results have been compared with other methods
and these results show the efficiency of the proposed
algorithm in restoring images that have been damaged
due to blur and noise.


\begin{thebibliography}{99}
\providecommand{\doi}[1]{DOI~\discretionary{}{}{}#1}
\bibitem{23}
Adam, Tarmizi, and Raveendran Paramesran. \textit{Hybrid non-convex second-order total variation with applications to non-blind image deblurring.} Signal, Image and Video Processing 14, no. 1 (2020): 115-123.

\bibitem{11}
Bioucas-Dias, Jos\'{e} M., M\'{a}rio AT Figueiredo, and Joao P. Oliveira. \textit{Adaptive total variation image deconvolution: A majorization-minimization approach.} In 2006 14th European Signal Processing Conference, pp. 1-4. IEEE, 2006.

\bibitem{31}
Cai, Jian-Feng, Jae Kyu Choi, Jingyang Li, and Ke Wei.
\textit{Image restoration: Structured low rank matrix
framework for piecewise smooth functions and beyond.}
Applied and Computational Harmonic Analysis 56 (2022): 26-60.

\bibitem{29}
Cai, Jian-Feng, Stanley Osher, and Zuowei Shen. \textit{Split Bregman methods and frame based image restoration.} Multiscale modeling $\&$ simulation 8, no. 2 (2010): 337-369.

\bibitem{18}
Chen, Scott Shaobing, David L. Donoho, and Michael A. Saunders. \textit{Atomic decomposition by basis pursuit.} SIAM review 43, no. 1 (2001): 129-159.

\bibitem{28}
De Oliveira, Edmundo Capelas, and Jos\'{e} Ant\'{o}nio Tenreiro Machado.
\textit{A review of definitions for fractional derivatives and integral.} Mathematical Problems in Engineering 2014 (2014).

\bibitem{36}
Dong, Bin, Hui Ji, Jia Li, Zuowei Shen, and Yuhong Xu. \textit{Wavelet frame based blind image inpainting.} Applied and Computational Harmonic Analysis 32, no. 2 (2012): 268-279.

\bibitem{16}
Dong, Fangfang, and Yunmei Chen. \textit{A fractional-order derivative based variational framework for image denoising}. Inverse Problems $\&$ Imaging 10, no. 1 (2016): 27.

\bibitem{17}
Dong, Fangfang, and Qianting Ma. \textit{Single image blind deblurring based on the fractional-order differential.} Computers $\&$ Mathematics with Applications 78, no. 6 (2019): 1960-1977.

\bibitem{15}
Esser, Ernie, Xiaoqun Zhang, and Tony Chan. \textit{A general framework for a class of first order primal-dual algorithms for TV minimization.} Ucla Cam Report 9 (2009): 67.

\bibitem{10}
Han, Bin. Framelets and wavelets. In Algorithms, Analysis, and Applications, Applied and Numerical Harmonic Analysis. Birkh\"{a}user xxxiii Cham, 2017.

\bibitem{8}
Han, Bin. \textit{Properties of discrete framelet transforms.} Mathematical Modelling of Natural Phenomena 8, no. 1 (2013): 18-47.


\bibitem{30}
He, Yiyang, Hongli Wang, Lei Feng, and Sihai You.
\textit{Motion-blurred star image restoration based on
multi-frame superposition under high dynamic and long
exposure conditions.} Journal of Real-Time Image Processing 18, no. 5 (2021): 1477-1491.


\bibitem{35}
Jing, Yu, Jianxin Liu, Zhaoxia Liu, and Hongju Cao. \textit{Fast edge detection approach based on global optimization convex model and split Bregman algorithm.} Diagnostyka 19 (2018).

\bibitem{6}
Kumar, Rohit, and Vishal Moyal. \textit{Visual image quality assessment technique using fsim.} International Journal of Computer Applications Technology and Research 2, no. 3 (2013): 250-254.

\bibitem{26}
Li, Bo, and Wei Xie. \textit{ Adaptive fractional differential
approach and its application to medical image enhancement.} Computers $\&$ Electrical Engineering 45 (2015): 324-335.

\bibitem{32}
Li, Dongming, Changming Sun, Jinhua Yang, Huan Liu,
Jiaqi Peng, and Lijuan Zhang.
\textit{Robust multi-frame adaptive optics image
restoration algorithm using maximum likelihood estimation with poisson statistics.} Sensors 17, no. 4 (2017): 785.

\bibitem{3}
Liu, Jingjing, Anqi Ni, and Guoxi Ni. \textit{A nonconvex l1(l1-l2) model for image restoration with impulse noise}. Journal of Computational and Applied Mathematics 378 (2020): 112934.
\bibitem{24}
Liu, Jingjing, Ruijie Ma, Xiaoyang Zeng, Wanquan Liu, Mingyu Wang, and Hui Chen.\textit{ An efficient non-convex total variation approach for image deblurring and denoising.} Applied Mathematics and Computation 397 (2021): 125977.
\bibitem{2}
Liu, Jun, Ting-Zhu Huang, Xiao-Guang Lv, and Jie Huang.
\textit{Restoration of blurred color images with impulse noise}.
 Computers $\&$ Mathematics with Applications 70, no. 6 (2015): 1255-1265.
\bibitem{7}
Lou, Yifei, and Ming Yan. \textit{Fast L1-L2 minimization via a proximal operator.} Journal of Scientific Computing 74, no. 2 (2018): 767-785.
\bibitem{27}
Love, Eric Russell. \textit{Fractional derivatives of
imaginary order.} Journal of the London Mathematical Society 2, no. 2 (1971): 241-259.
\bibitem{newr1}
Mistakidis, Euripidis S., and Georgios E. Stavroulakis. Nonconvex optimization in mechanics: algorithms, heuristics and engineering applications by the FEM. Vol. 21. Springer Science $\&$ Business Media, 2013.
\bibitem{22}
Oliveira, Joao P., Jos\'{e} M. Bioucas-Dias, and Mário AT Figueiredo. \textit{Adaptive total variation image deblurring: a majorization–minimization approach.} Signal processing 89, no. 9 (2009): 1683-1693.

\bibitem{21}
Rudin, Leonid I., Stanley Osher, and Emad Fatemi. \textit{Nonlinear total variation based noise removal algorithms.} Physica D: nonlinear phenomena 60, no. 1-4 (1992): 259-268.

\bibitem{19}
Sajjad, Muhammad, Irfan Mehmood, Naveed Abbas, and Sung Wook Baik. \textit{Basis pursuit denoising-based image superresolution using a redundant set of atoms.} Signal, Image and Video Processing 10, no. 1 (2016): 181-188.

\bibitem{5}
Sara, Umme, Morium Akter, and Mohammad Shorif Uddin.
\textit{Image quality assessment through FSIM, SSIM, MSE and PSNR-a comparative study.}
Journal of Computer and Communications 7, no. 3 (2019): 8-18.

\bibitem{9}
Shen, Yi, Bin Han, and Elena Braverman. \textit{Image inpainting from partial noisy data by directional complex tight framelets.} The ANZIAM Journal 58, no. 3-4 (2017): 247-255.
\bibitem{33}
Setzer, Simon. \textit{Split Bregman algorithm, Douglas-Rachford splitting and frame shrinkage.} In International Conference on Scale Space and Variational Methods in Computer Vision, pp. 464-476. Springer, Berlin, Heidelberg, 2009.


\bibitem{1}
Tao, Min, Junfeng Yang, and Bingsheng He. \textit{Alternating direction
 algorithms for total variation deconvolution in image reconstruction}.
 TR0918, Department of Mathematics, Nanjing University (2009).

\bibitem{4}
Wang, Zhou, Alan C. Bovik, Hamid R. Sheikh, and Eero P. Simoncelli.
\textit{Image quality assessment: from error visibility to structural similarity.}
 IEEE transactions on image processing 13, no. 4 (2004): 600-612.

\bibitem{25}
Yang, Jingjing, Yingpin Chen, and Zhifeng Chen. \textit{Infrared Image Deblurring via High-Order Total Variation and Lp-Pseudonorm Shrinkage.} Applied Sciences 10, no. 7 (2020): 2533.

\bibitem{12}
Zhang, Yi, Weihua Zhang, Yinjie Lei, and Jiliu Zhou. \textit{Few-view image reconstruction with fractional-order total variation.} JOSA A 31, no. 5 (2014): 981-995.
\bibitem{14}
Zhu, Mingqiang, and Tony Chan. \textit{An efficient primal-dual hybrid gradient algorithm for total variation image restoration.} UCLA CAM Report 34 (2008): 8-34.
\bibitem{20}
Zifan, Ali, and Panos Liatsis. \textit{Medical image deblurring via lagrangian pursuit in frame dictionaries.} In 2011 Developments in E-systems Engineering, pp. 86-91. IEEE, 2011.
\bibitem{34}
Zou, Jian, Haifeng Li, and Guoqi Liu. \textit{Split Bregman algorithm for structured sparse reconstruction.} IEEE Access 6 (2018): 21560-21569.
\end{thebibliography}
\end{document}